\RequirePackage{fix-cm}
\documentclass[smallextended]{svjour3}       
\smartqed  
\usepackage{graphicx}
\usepackage{graphicx}
\usepackage{courier}
\usepackage{epstopdf}
\usepackage{color}
\usepackage{csquotes}
\usepackage{enumerate}
\usepackage{amsmath}
\usepackage{amsfonts}
\usepackage{amssymb}
\usepackage[subnum]{cases}
\usepackage{booktabs} 
\usepackage{pifont}

\usepackage{amssymb}
\usepackage{pifont}
\newcommand{\cmark}{\ding{51}}%
\newcommand{\xmark}{\ding{55}}%

\usepackage {tikz}
\usetikzlibrary {positioning}
\definecolor {processblue}{cmyk}{0.96,0,0,0}
\usepackage{lipsum,adjustbox}

\usepackage[subnum]{cases}
\usepackage{color,soul}
\usepackage{subfig}

\newtheorem{axiom}{Axiom}
   
\begin{document}

\title{A Generic Axiomatic Characterization of Centrality Measures in Social Network
}


\author{Sambaran Bandyopadhyay         \and
        M. Narasimha Murty               \and
        Ramasuri Narayanam
}


\institute{F. Author \at
              IBM Research \\
              \email{sambband@in.ibm.com}           
           \and
           S. Author \at
              Indian Institute of Science, Bangalore \\
              \email{mnm@csa.iisc.ernet.in}
           \and
           T. Author \at
           	  IBM Research \\
           	  \email{nrsuri@gmail.com}             
}

\date{Received: date / Accepted: date}

\maketitle

\begin{abstract}
Centrality is an important notion in complex networks; it could be used to characterize how influential a node or an edge is in the network. It plays an important role in several other network analysis tools including community detection. Even though there are a small number of axiomatic frameworks associated with this notion, the existing formalizations are not generic in nature. In this paper we propose a generic axiomatic framework to capture all the intrinsic properties of a centrality measure (a.k.a. centrality index). We analyze popular centrality measures along with other novel measures of centrality using this framework. We observed that none of the centrality measures considered satisfies all the axioms.

\keywords{Centrality \and Axiomatic Framework \and Networks}

\end{abstract}


\section{Introduction}

Analyzing complex network is important for dealing with several real-world applications. Online social networks (e.g., Facebook or Flickr), collaboration networks, email networks, trading networks, genetic networks and  R \& D networks \cite{Easley:2010,Brandes:2005} are some important examples of complex networks. Social networks are social structures made up of groups or communities of individuals and connections among these individuals. It  is convenient to represent such networks using graphs, where nodes represent entities in the networks and edges represent the connections among these entities; the underlying hope is that well-known properties of graphs can be exploited in the analysis of the networks. A significant amount of work on social network analysis is devoted to understanding the role played by nodes/edges, with respect to either their structural placement in the network or their behavioral influence over others in the network. To this end, it is important to rank nodes/edges in a given network based on either their positional power or their behavioral influence. Traditionally, nodes with high rank are referred to as {\em influential nodes}. 

We present two motivating examples to understand the importance of finding the influential nodes in real world social networks. The first example deals with the diffusion of information in social networks wherein it is required to initially target a few influential individuals in the network who will trigger a massive cascade of influence through which friends will recommend the product to other friends, and many individuals will ultimately try it. The second example deals with co-authorship social networks wherein it may be of interest to find the most prolific researchers since they are most likely to be the trend setters for breakthroughs. The common goal in these two example settings is to find a set of influential nodes given a well defined context in the social network. To achieve this objective, there exist several well known ranking mechanisms in the literature, ranging from the well known centrality measures (a.k.a. centrality indices) from social sciences such as degree centrality, closeness centrality, clustering coefficient, and betweenness centrality \cite{faster_betweenness} to Google PageRank \cite{Brin:98}. A pragmatic reason behind the existence of such influential nodes and edges is the community structure exhibited by most of the networks; an influential node generally binds together nodes in the community or forms a bridge between two communities.

Of late, a large variety of centrality measures \cite{Easley:2010,Brandes:2005} have been proposed by the research community whenever the existing measures in the literature are proved to be inadequate to satisfactorily serve the needs of emerging real-life applications. Though these centrality measures offer new insights, the lack of a theoretical underpinning makes it difficult to choose the right centrality measure for a given context. Towards this end, there exists some effort in the literature \cite{skibski_2017},\\ \cite{jair_altman_2008}, \cite{boldi_2014},\\ \cite{skibski2016attachment} in terms of developing axiomatic frameworks to better understand the properties of these centrality measures. However, these theoretical explorations are limited to only certain specific scenarios. To the best of our knowledge, there is no generic theoretical framework to understand the structural properties exhibited by the class of centrality measures in networks. In this paper, we attempt to address this research gap by proposing a novel axiomatic framework.

\subsection{Key Contributions}
Following are the contributions made in this paper.
\begin{itemize}
    \item Existing axiomatic approaches are not generic. They have either been designed to characterize a particular centrality measure, or applicable to networks of specific structures. We propose a novel axiomatic framework which captures all the intrinsic properties that a centrality measure is expected to satisfy.
    \item We analyze well-known centrality measures in the proposed framework. Surprisingly most of them do not satisfy one or more of the fundamental properties of our framework.
    \item Also we have designed some intuitively appealing centrality measures that satisfy some of the axioms of the framework. We believe that our framework would help researchers to develop and analyze the quality of new centrality measures for different types of networks.
\end{itemize}



\section{Notation and Terminology} \label{sec:def}
In this section, we provide the necessary background and notation used in the subsequent sections of the paper. Social networks are typically represented by graphs. In this work, we have concentrated mainly on the representation of social networks using undirected and unweighted graphs. Following are some of the useful definitions borrowed from graph theory. Details can be found in any standard graph theory text book such as \cite{diestel2000graph}.

A graph, or a network, is a pair, $G = (V,E)$, where $V$ is the set of $n=|V|$ nodes (a.k.a. vertices), and $E$ is the set of edges. We consider only the set of undirected graphs in this paper. Sometime the set of nodes $V$ in the graph $G$ is also denoted as $V(G)$ and similarly the set of edges $E$ is denoted by $E(G)$. The set of all possible graphs with nodes $V$ is denoted by $\mathcal{G}^V$. An edge which is incident on nodes $u$ and $v$ is denoted by $\{u,v\}$. We say that two vertices $u$ and $v$ are adjacent to each other if $\{u,v\} \in E$. For an undirected graph, $\{u,v\}$ = $\{v,u\}$. The degree of a node $u$ is defined as:
\begin{align}
degree_G(u) = |\{\{u,v\} \in E : v \in V\}|
\end{align}
where $|.|$ denotes the cardinality of a set. A node $u$ in a graph $G$ is called an isolated node if $degree_G(u) = 0$.

A path $p=(u_1,\cdots,u_k)$ is a sequence of nodes such that any two consecutive nodes are connected by an edge. The length of a path is the number of edges in it. The distance $dist(u,v)$ between any two nodes $u$ and $v$ in the graph $G$ is the length of the shortest path between them. If there is no path between $u$ and $v$, it is assumed that $dist(u,v) = \infty$. The set of shortest paths between nodes $u$ and $v$ is denoted by $\Pi_s(u,v)$. A connected component of a graph is a subset of nodes such that any two nodes in the subset can be reached from one another by a path. $K(G)$ is the partition of $V$ into disjoint sets of nodes where each node induces a maximal connected subgraph in $G$, and $K_u(G)$ is the connected component containing the node $u$ in $G$. By $G[K_u(G)]$, we denote the subgraph of $G$ which contains the nodes and the edges within the component $K_u(G)$.

For any vertex $u \in V$, $N_G(u)$ denotes the set of neighbor nodes of the vertex $u$ in $G$. So, $N_G(u) = \{v \in V \; : \; \{u,v\} \in E\}$. Similarly, for any vertex $u$ in the graph $G$, we define the $h$-hop neighbor set of $u$ as $H_G^h(u) = \{v \in V : dist(u,v)=h\}$. Note that $N_G(u) = H_G^1(u)$. We extend the idea of $h$-hop neighbor sets even for a subset of nodes in the graph. For $u$ and $v$ $\in V$, the $h$-hop neighbor set of $(u,v)$ is defined as $H_G^h(u,v) = \{z \in V : min(dist(u,z), dist(v,z)) = h \}$. In case $\{u,v\} \in E$, this also denotes the $h$-hop neighbor of the edge $\{u,v\}$.

We call two graphs $G$ and $H$ to be isomorphic to each other if there is a bijection between the vertex sets of $G$ and $H$ $f : V(G) \rightarrow V(H)$ such that any two vertices $u$ and $v$ of $G$ are adjacent in $G$ if and only if $f(u)$ and $f(v)$ are adjacent in $H$. Adjacency matrix $A$ of the graph $G$ is a $n \times n$ matrix defined as:
\begin{align}
A_{i,j} = 
\begin{cases}
1 \; , \; if \;\; \{i,j\} \in E \\
0 \; , \; otherwise  
\end{cases}
\end{align} 

The pair $<\lambda,x>$ is an eigenpair of $A$ if $Ax = \lambda x$, where $\lambda$ is the eigenvalue of $A$ and $x\;(\neq 0)$ is the corresponding eigenvector.

Typically a social network is represented in the form of a graph. A centrality index or centrality measure, $F : \mathcal{G}^V \rightarrow \mathbb{R}_+^V$ , is a function that assigns to every node a non-negative real number which reflects its importance in the network. Hence the centrality of a node $u$ in the graph $G$ is denoted by $F_{u}(G)$. Typically, larger the value of this index, more important or central the node is. Some of the well known centrality measures are Degree Centrality, Closeness Centrality, Betweenness Centrality and Eigenvector Centrality \cite{spizzirri2011justification,skibski2016attachment}. 
\section{Background}
Centrality of nodes and edges is important in social network analysis\\ \cite{koschutzki2005centrality}. A node is central if several other nodes are connected to it or if it is connected to two or more different communities. Similarly an edge is central if it acts as a bridge between two communities. Based on this notion, some of the measures of centrality \cite{freeman1978centrality} which are popularly used are:
\begin{enumerate}[(a)]
\item {\em Degree centrality:} A node is central if it is better connected to other nodes or its {\em degree is high}.
\item {\em Closeness centrality:} A node is central if it is closer to the rest of nodes or the {\em average distance} between the node and the other nodes is small.
\item {\em Betweenness centrality:} A node is central if it is {\em between} two subsets of nodes or communities.
\item {\em Eigenvector centrality:} It involves a recursive characterization. A node is central if it is linked to other central nodes. Formally, if $\lambda$ is the largest eigenvalue of $A$ and $x$ is the corresponding eigenvector, then node $i$ is more central than any other node $j$ if the $i^{th}$ component of $x$ is larger than its $j^{th}$ component.
\end{enumerate}
%
\begin{figure}[h]
\includegraphics[height=6cm,width=12cm]{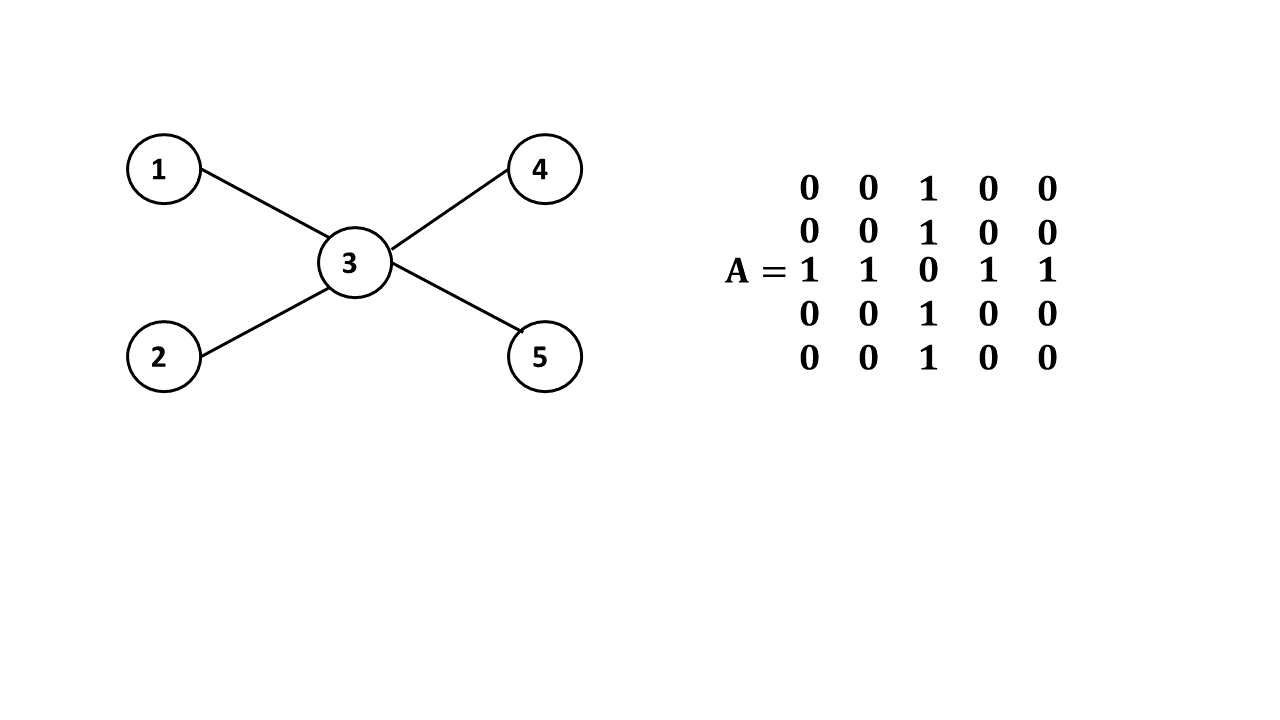}
\vspace*{-2cm}
\begin{center}
\caption{Example Network and Its Adjacency Matrix}
\label{fig:centra}
\end{center}
\end{figure}
Formal definitions of these centrality measures are given during their analysis in Section \ref{sec:sat}. We illustrate these centrality measures using a simple example network shown in Figure~\ref{fig:centra}.
\begin{table}
\begin{center}
\resizebox{0.8\textwidth}{!}{%
\begin{tabular}{c c c c c}
Node&Degree&Closeness&Betweenness&Eigenvector\\
\hline
&&&&\\
1&$1$&$\frac{5}{2}$&$0$&$\frac{1}{2  \sqrt{2}}$\\
&&&&\\
2&$1$&$\frac{5}{2}$&$0$&$\frac{1}{2  \sqrt{2}}$\\
&&&&\\
3&$4$&$4$&$6$&$\frac{1}{\sqrt{2}}$\\
&&&&\\
4&$1$&$\frac{5}{2}$&$0$&$\frac{1}{2  \sqrt{2}}$\\
&&&&\\
5&$1$&$\frac{5}{2}$&$0$&$\frac{1}{2  \sqrt{2}}$
\end{tabular}
}
\end{center}
\caption{Centrality values of the nodes in the example network in Figure \ref{fig:centra}.}
\label{table:four}
\end{table}

\begin{enumerate}[(a)]
\item In the second column of Table \ref{table:four}, the degrees of the five nodes are listed. Node 3 has the highest degree of 4 and it is the most central node.
\item In terms of closeness, node 3 is again central as its average distance to the remaining four nodes is 1; other nodes have larger average distance values. Here each edge in the network contributes to a distance of 1 unit between the end nodes.
\item Only node 3 has a non-zero value for betweenness. Every path between a pair of nodes other than 3
passes through node 3.
\item The last column depicts the eigenvector of the adjacency matrix corresponding to the eigenvalue 2.
Here also node 3 has the highest centrality value. 
\end{enumerate}
We will analyze each of them with respect to our axiomatic framework in the subsequent sections.

\section{Related Work}
In this section we examine the existing literature on centrality measures including formal treatments based on axiomatic characterizations. We also distinguish our work from the existing axiomatic approaches characterizing centrality in social networks.

Axiomatic frameworks have been used in different domains of computer science and economics, where the end goals are intuitively clear but not mathematically rigorous enough. In a typical axiomatic framework, axioms are proposed and used to capture the intrinsic properties of the underlying concept. For example, clustering \cite{jain1999data} is an extensively used tool in data mining and machine learning, yet it is an ill-posed problem. So research has been done to capture the intrinsic properties of clustering in the form of axioms \cite{kleinberg2002impossibility,zadeh2009uniqueness,bandyopadhyay2016axioms}.

 Similar approaches have been adopted in other domains such as social choice theory \cite{kelly2014arrow}, ranking and diversification \cite{gollapudi2009axiomatic}, and computational sustainability and dynamic pricing\\ \cite{bandyopadhyay2016axiomatic}. Centrality has also been used to deal with other tasks associated with the social network analysis. For example, betweenness centrality is used in community detection \cite{Jure_2014}.
 
 There are other axiomatic frameworks for centrality in networks. First we want to discuss the framework presented in \cite{boldi_2014}. The authors have proposed axioms to characterize the effect of size, density and addition of an edge in a network represented in the form of a directed graph. But they have mainly focused on the graph made by a $k$-clique and a directed $p$-cycle and compare the centrality of the nodes within that graph. Hence the scope of the framework is limited to this particular structure of the network.
 
 Second, in \cite{skibski2016attachment}, authors characterized degree centrality and attachment centrality (based on Myerson value \cite{myerson1977graphs}) in the proposed axiomatic framework. Similar to us, they also represent networks in the form of an undirected graph. But unfortunately, some of the proposed axioms in their framework are particular to degree or attachment centrality and loose significance in the broader context. For example, there is not much justification to assume that adding an edge would exactly have the same effect on the two incident nodes as stated in the Fairness Axiom. Similarly a generic reason for using a linear sum in Gain-loss Axiom has not been mentioned in the paper. We discuss and compare with the Monotonicity Axiom in more detail in Section \ref{sec:axioms}. 
 
 Research has also been done on the axiomatic foundations for ranking systems in a directed graph \cite{jair_altman_2008} in the context of Page Rank and voting ranking systems. Recently an axiomatic framework for game-theoretic network centralities \cite{skibski_2017} is proposed, where the authors establish a link between the game-theoretic centrality measures and classical centrality measures. Naturally the inherent set up of their analysis is significantly different from our framework.
 
 
 To summarize, the existing frameworks are not generic enough to deal with a variety of centrality measures. We propose
 a generic framework that can be used to analyze all the popular centrality measures. In addition, we consider some more intuitively appealing centrality measures. Our observation is that none of the centrality measures considered is able to satisfy all our axioms.
 \nocite{sabidussi1966centrality}

\section{Proposed Axiomatic Framework} \label{sec:axioms}
In this section, we capture all the intrinsic properties of a centrality measure using a set of axioms.

First axiom in our framework is a fundamental property of many graph theoretic measures. In the context of our work, it ensures that the centrality measure of a node in a graph should depend only on the structure of the graph. Hence if two graphs are isomorphic to each other, the centrality values of the corresponding nodes in the two graphs is the same.
\begin{axiom}
\label{ax:iso}
\textbf{Isomorphic invariant:} If two graphs $G$ and $H$ are isomorphic and $f : V(G) \rightarrow V(H)$ being the structure preserving bijection, then $F_v(G) = F_{f(v)}(H)$, $\forall v \in V(G)$.
\end{axiom}
Centrality measure is explicitly constrained to be a structural index in\\ \cite{koschutzki2005centrality},  which by definition satisfies Axiom \ref{ax:iso}.  We put this property in the form of an axiom to provide a rigorous characterization of centrality.

Centrality of a node depends on the way the node is connected to the other nodes in the graph. Naturally it should depend only on the maximally connected component in which it is located. So the centrality of the node would remain unchanged even if we only consider the subgraph induced by the corresponding maximally connected component.
\begin{axiom}
\label{ax:loc}
\textbf{Locality:} For every graph $G = (V,E)$ and every node $v \in V$, the centrality of $v$ depends solely on $G[K_v(G)]$. That is,
\begin{align*}
F_v(G) = F_v(G[K_v(G)])
\end{align*}
\end{axiom}
Axiom \ref{ax:loc} is also present in \cite{skibski2016attachment}. It is a natural extension of centrality measures from connected to disconnected graphs in general.

In a graph, isolated nodes are completely separated from all other nodes, and hence they do not play any role in the connectivity of the network. Naturally the centrality of an isolated node would be the least possible value consistent with the functional definition of a centrality measure as given in Section \ref{sec:def}. Hence it leads us to the following axiom.
\begin{axiom}
\label{ax:min}
\textbf{Isolated Minima:} For any graph $G= (V,E)$, $F_v(G) = 0$, where $v$ is a isolated node in $G$.
\end{axiom}

The first three axioms of our framework are basic properties of a centrality measure. Further, the edges of a network play an important role to determine the centrality of a node in the network. So we capture the effect of edges in the centrality of a node in the next two axioms. As the connectivity of a node in the network increases, its centrality should also increase. One way to increase the connectivity of a node is by having more edges incident on it. Following axiom captures the affect, of having a new edge, on the two end point nodes of the network.
\begin{axiom}
\label{ax:mon}
\textbf{Edge Monotonicity:} For every graph $G = (V,E)$, and two nodes $u$ and $v$ such that $\{u,v\} \notin E$, $F_u((V, E \cup \{u,v\})) > F_u((V,E))$ and $F_v((V, E \cup \{u,v\})) > F_v((V,E))$.
\end{axiom}
It is important to understand that adding an edge can have different impacts on the other nodes of the network. The centrality of some other nodes might increase or decrease, but it is fair to assume that the nodes which are getting connected because of the new edge would always get benefited, and become more central to the network. It is also worthwhile to mention that in \cite{skibski2016attachment}, authors have proposed an axiom which says that, adding an edge does not decrease the centrality of any node in the graph. Clearly this is not a valid assumption. For example, in a road traffic network, as a new road (edge) is laid between two terminal places, say X and Y (nodes), a portion of the traffic (importance or centrality) from some nearby place may deviate and start to follow the new road. Naturally the importance of X and Y increases while it decrease for some other places.

In the last axiom, we examine the effect of adding an edge on the immediate two nodes which get joined by the new edge. The natural question to ask here is, whether it is possible to generalize the effect on the other nodes of the network. Typically in social science and psychology ~\cite{Ana_2001}, the effect of any change in the network can have significant impact on the immediate neighbors. This effect diminishes as the distance from the point of impact increases. This observation can be formally stated in the form of the following axiom.
\begin{axiom}
\label{ax:hop}
\textbf{Diminishing Impact:}Let us add an edge $\{u,v\}$ to $E$ of the graph $G=(V,E)$, where $u, v \in V$ and $\{u,v\} \notin E $, and consider the new graph as $G'=(V, E \cup \{u,v\})$. Then for any two non-negative integers $h$ and $\bar{h}$ such that $h < \bar{h}$ and there exist nodes $z_h \in H_G^{h}(u,v)$ and $z_{\bar{h}} \in H_G^{\bar{h}}(u,v)$, then $|F_{z_h}(G') - F_{z_h}(G)| > |F_{z_{\bar{h}}}(G') - F_{z_{\bar{h}}}(G)|$.
\end{axiom}
It is important to note that, the effect of adding the new edge $\{u,v\}$ as stated in the above axiom, can increase or decrease the centrality measure of any other node $z$ in the network, but the absolute value of this change would diminish over the distance from the edge $\{u,v\}$.

Also combining Axiom \ref{ax:mon} and Axiom \ref{ax:hop}, we can conclude that maximum change of centrality would occur for the nodes $u$ and $v$, as both of them belong to the $0$-hop neighborhood of $(u,v)$, and also the change is positive as the centrality values of both nodes increase. 

In the last two axioms, we discuss the effect of adding an edge to the network. But so far we have not discussed how the centrality measures of different nodes would compete with each other within the same network. A node is clearly central if it is connected to more number of nodes which are central themselves. For example in a co-authorship network, a researcher who has written research papers with some of the leading researchers in the domain is also assumed to be influential in the community. On the other hand, a node in a very small component of the network may not have a central role compared to a node which is at the center of a much bigger component of the network. The idea is captured in the following axiom.
\begin{axiom}
\label{ax:struc}
\textbf{Structural Consistency:} Consider a graph $G=(V,E)$ and any two vertices $u, v \in V$ such that $|N_G(u)| \geq |N_G(v)|$. If there exists a subset $\bar{N}_G(u) \subseteq N_G(u)$ with $|\bar{N}_G(u)| = |N_G(v)|$, such that there is a bijection $h$ which attaches each vertex $a \in \bar{N}_G(u)$ to a unique vertex $h(a) \in N_G(v)$ so that $F_a(G) > F_{h(a)}(G)$, then $F_u(G) > F_v(G)$. 
\end{axiom}



As one can see, we have not considered all the corner cases in the design of a centrality measure in our axiomatic framework. For example, in Axiom \ref{ax:mon}, we do not comment on the effect of adding an edge to the centrality of any randomly taken node in the network. In Axiom \ref{ax:hop}, we do not impose any condition on the change of centrality of the nodes which are in the same hop neighborhood to the newly added edge. Similarly in Axiom \ref{ax:struc}, we carefully avoid the case when $|N_G(u)| < |N_G(v)|$ and such a bijection exists from $N_G(u)$ to $\bar{N}_G(v) \subseteq N_G(v)$ with the same conditions. We felt that the results of these cases can differ significantly from one use case to another. So we leave these cases to be taken care of by the designer of the centrality measure based on the properties of the specific network under consideration. Thus our axiomatic framework is a generic one which is applicable to different types of networks, and also it gives freedom to the network designers to impose extra conditions for the particular use case they are handling.

\section{Satisfiability of the Axioms} \label{sec:sat}
In this section we will characterize different measures of centrality with respect to our axiomatic framework. We also propose novel centrality measures in the process of this analysis. We state some important observations in the form of lemmas in this section. We skipped the proof of some of these lemmas when it is trivial. 

\subsection{Uniform Centrality}
To begin with, let us define a simple centrality measure, called uniform centrality (UC), as below:
\begin{align}
UC_u(G) = \beta
\end{align}
where $\beta$ is a fixed number. Uniform centrality can be used as a simple prior to design advanced centrality measures such as Page Rank \cite{clauset2013network}.

\begin{lemma}
Uniform centrality satisfies Axioms \ref{ax:iso} and \ref{ax:loc}, but does not satisfy Axioms \ref{ax:min}, \ref{ax:mon}, \ref{ax:hop} and \ref{ax:struc}.
\end{lemma}
The proof is immediate as the uniform centrality assigns the same score to all the nodes in the network.

\subsection{Degree Centrality}
Degree centrality is defined as:
\begin{align}
	DC_u(G) = Degree_G(u)
	\label{eq:WDC}
\end{align}

When the graph $G$ is fixed, we may omit $G$ from the notation of centrality measure for the sake of brevity. So for a fixed $G$, $DC_u$ actually means $DC_u(G)$.  

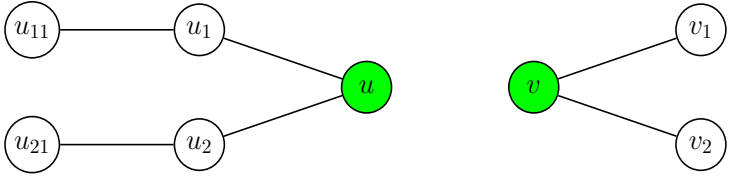
\begin{figure}
\centering
\resizebox{3.8in}{0.9in}{
\begin {tikzpicture}[auto, node distance =1 cm and 3cm, on grid, every loop/.style={},thick, state/.style={circle,draw,minimum size=3em,font=\sffamily\Large\bfseries}]\node[state, fill=green] (U) {$u$};

\node[state] (U1) [above left=of U] {$u_1$};
\node[state] (U2) [below left=of U] {$u_2$};
\node[state] (U11) [left=of U1] {$u_{11}$};
\node[state] (U21) [left=of U2] {$u_{21}$};

\node[state, fill=green] (V) [right=of U] {$v$};
\node[state] (V1) [above right=of V] {$v_1$};
\node[state] (V2) [below right=of V] {$v_2$};

\path (U) edge (U1);
\path (U) edge (U2);
\path (U1) edge (U11);
\path (U2) edge (U21);

\path (V) edge (V1);
\path (V) edge (V2);
\end{tikzpicture}
}
\caption{\text{Degree Centrality does not satisfy Axiom \ref{ax:struc}}. Here $DC_u =DC_v =2$, but $DC_{u_1} = DC_{u_2} = 2 > 1 = DC_{v_1} = DC_{v_2}$}
\label{fig:DC}
\end{figure}

\begin{lemma}
Degree Centrality satisfies Axioms \ref{ax:iso}, \ref{ax:loc}, \ref{ax:min} and \ref{ax:mon}.
\end{lemma}
We skip the proof as this is immediate.

\begin{lemma}
Degree centrality does not satisfy Axioms \ref{ax:hop} and \ref{ax:struc}.
\end{lemma}
\begin{proof}
Clearly adding a new edge $\{u,v\}$ only changes the centrality of the two immediate nodes $u$ and $v$, the degree centrality of all other nodes in the network remains the same. Hence degree centrality does not satisfy Axiom \ref{ax:hop}.

For the other part of the proof, we give a counter example in Figure \ref{fig:DC}. Here $DC_{u_1}>DC_{u_1}$ and $DC_{u_2}>DC_{v_2}$, but $DC_{u} = DC_{v}$. Hence degree centrality fails to satisfy Axiom \ref{ax:struc}.
\end{proof}

Degree centrality is a simple measure of centrality which only captures the local behavior of a node in the network. Naturally it fails to satisfy the last two axioms of our framework.

\subsection{Closeness centrality} 

Closeness centrality is defined as:
\begin{align}
	CC_u(G) = \sum\limits_{w \in V \setminus \{u\}} \frac{1}{dist(u,w)}
	\label{eq:WDC}
\end{align}

\begin{lemma}
Closeness centrality satisfies Axioms \ref{ax:iso}, \ref{ax:loc}, \ref{ax:min} and \ref{ax:mon}.
\end{lemma}
Again we omit the proof as it is straightforward. One can see that adding an edge would always strictly increase the closeness centrality of both the incident nodes as they are brought closer because of the new edge.

\begin{figure}
\centering
\resizebox{2.9in}{1.5in}{
\begin {tikzpicture}[auto, node distance =1 cm and 3cm, on grid, every loop/.style={},thick, state/.style={circle,draw,minimum size=2em,font=\sffamily\Large\bfseries}]

\node[state] (U) {$u$};
\node[state, fill=green] (Z1) [left=of U] {$z_1$};
\node[state] (V) [right=of U] {$v$};
\node[state] (W) [above=of U] {$w$};
\node[state] (X) [above=of V] {$x$};
\node[state] (Y) [below=of U] {$y$};
\node[state, fill=green] (Z2) [below=of Y] {$z_2$};

\path (U) edge (Z1);
\path (Z1) edge (W);
\path (W) edge (X);
\path (X) edge (V);
\path (U) edge (Y);
\path (Y) edge (Z2);

\path (U) edge (V) [dashed];

\end{tikzpicture}
}
\caption{Closeness Centrality does not satisfy Axiom \ref{ax:hop}. Change of centrality for $z_2$ is more that that for $z_1$ because of adding the new edge $\{u,v\}$.}
\label{fig:CC}
\end{figure}
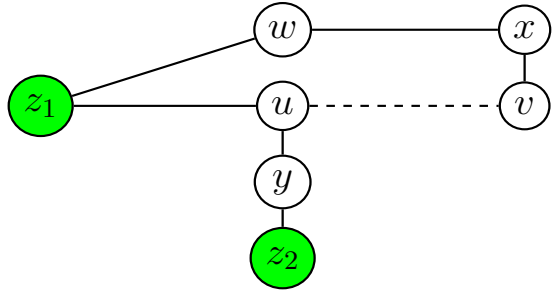

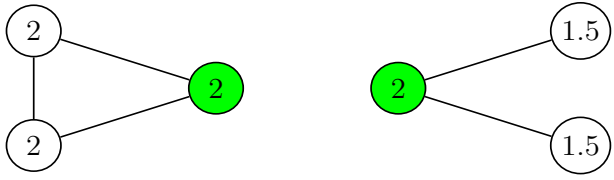
\begin{figure}
\centering
\resizebox{3.2in}{0.9in}{
\begin {tikzpicture}[auto, node distance =1 cm and 3cm, on grid, every loop/.style={},thick, state/.style={circle,draw,minimum size=3em,font=\sffamily\Large\bfseries}]

\node[state, fill=green] (U) {$2$};
\node[state] (U1) [above left=of U] {$2$};
\node[state] (U2) [below left=of U] {$2$};

\node[state, fill=green] (V) [right=of U] {$2$};
\node[state] (V1) [above right=of V] {$1.5$};
\node[state] (V2) [below right=of V] {$1.5$};

\path (U) edge (U1);
\path (U) edge (U2);
\path (U1) edge (U2);
\path (V) edge (V1);
\path (V) edge (V2);

\end{tikzpicture}
}
\caption{Closeness Centrality does not satisfy Axiom \ref{ax:struc}. Centrality values are labeled within the nodes in this graph. According to Axiom \ref{ax:struc}, the green node on the left component should have been strictly more central than that of the green node on the right component.}
\label{fig:BC1}
\end{figure}

\begin{lemma} \label{lemma:CC}
Closeness Centrality does not satisfy Axiom \ref{ax:hop}.
\end{lemma}
\begin{proof}
We prove this lemma by constructing a scenario where CC fails to satisfy Axiom \ref{ax:hop}.
Let us assume the graph $G$ is without the edge $(u,v)$ and the graph $G'$ is obtained by adding the edge $(u,v)$ to the graph $G$, as shown in Figure \ref{fig:CC}. Now it can be easily seen that,
$CC_{z_1}(G') - CC_{z_1}(G) = 1/2 - 1/3 = 1/6$ and $CC_{z_2}(G') - CC_{z_2}(G) = (1/4 -1/5) + (1/3 - 1/6) = 1/20 + 1/6 > 1/6$.

But $z_1 \in H_G^1(u,v)$  and $z_2 \in H_G^2(u,v)$. Hence proved.
\end{proof}

\begin{lemma} \label{lemma:CC}
Closeness Centrality does not satisfy Axiom \ref{ax:struc}.
\end{lemma}
\begin{proof}
Consider Figure \ref{fig:BC1}, the closeness centrality values (labeled inside the nodes) of the two green colored nodes are the same, but the closeness centrality of the neighbors of the node on the left hand side are strictly greater than those on the right hand side. Thus there exists a bijection as stated in Axiom \ref{ax:struc}. Hence this is a counterexample which shows that closeness centrality does not satisfy Axiom \ref{ax:struc}.
\end{proof}

\subsection{Betweenness Centrality}

Betweenness centrality is defined as:
\begin{align}
	BC_u(G) = \sum\limits_{\substack{s,t \in K_u(G) \\ s \neq t \neq u}} \frac{|p \in \Pi_s(s,t) : u \in p|}{|\Pi_s(s,t)|}
	\label{eq:WDC}
\end{align}

\begin{lemma}
Betweenness centrality satisfies Axiom \ref{ax:iso}, \ref{ax:loc} and \ref{ax:min}.
\end{lemma}
We skip the proof as it is trivial.

\begin{figure}
\centering
\resizebox{2.4in}{0.9in}{
\begin {tikzpicture}[auto, node distance =1 cm and 3cm, on grid, every loop/.style={},thick, state/.style={circle,draw,minimum size=3em,font=\sffamily\Large\bfseries}]

\node[state] (W) {$w$};
\node[state, fill=green] (U) [above left=of W] {$u$};
\node[state, fill=green] (V) [below left=of W] {$v$};
\node[state] (X) [right=of W] {$x$};

\path (W) edge (U);
\path (W) edge (V);
\path (W) edge (X);
\path (W) edge (X);
\path (U) edge (V)[dashed];

\end{tikzpicture}
}
\caption{Betweenness Centrality does not satisfy Axiom \ref{ax:mon}. Centrality of both the nodes $u$ and $v$ remain same even after adding the new edge $\{u,v\}$.}
\label{fig:BCmon}
\end{figure}
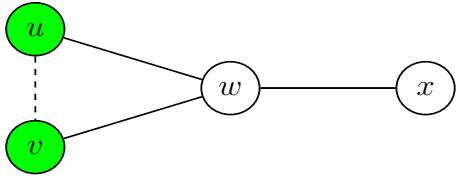

\begin{lemma}
Betweenness Centrality does not satisfy Axiom \ref{ax:mon}.
\end{lemma}
\begin{proof}
We give a counterexample as shown in Figure \ref{fig:BCmon}. Here we consider the graph $G$ without the edge $\{u,v\}$ and the graph $G'$ is obtained after adding the edge $\{u,v\}$. As one can check easily, $BC_u(G)=BC_u(G')$ and $BC_v(G)=BC_v(G')$. Thus it violets Axiom \ref{ax:mon}.
\end{proof}

\begin{lemma}
Betweenness Centrality does not satisfy Axiom \ref{ax:hop}.
\end{lemma}
\begin{proof}
Again we give a counterexample in this case.
Consider graphs $G$ and $G'$ as stated in the proof of Lemma \ref{lemma:CC}, and shown in Figure \ref{fig:CC}. Now consider nodes $w$ and $y$. Clearly, $BC_{w}(G) - BC_{w}(G') = 8 - 1 = 7$ and $BC_{y}(G) - BC_{y}(G') = 5 - 5 = 0 < 7$.

But $w \in H_G^2(u,v)$  and $y \in H_G^1(u,v)$. Hence proved.
\end{proof}

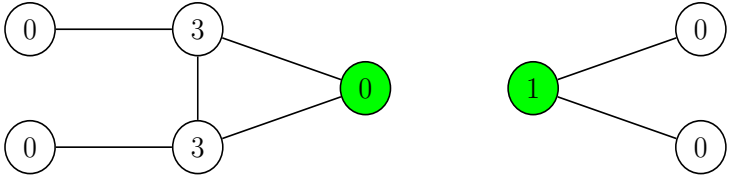
\begin{figure}
\centering
\resizebox{3.8in}{0.9in}{
\begin {tikzpicture}[auto, node distance =1 cm and 3cm, on grid, every loop/.style={},thick, state/.style={circle,draw,minimum size=3em,font=\sffamily\Large\bfseries}]

\node[state, fill=green] (U) {$0$};
\node[state] (U1) [above left=of U] {$3$};
\node[state] (U2) [below left=of U] {$3$};
\node[state] (U11) [left=of U1] {$0$};
\node[state] (U21) [left=of U2] {$0$};

\node[state, fill=green] (V) [right=of U] {$1$};
\node[state] (V1) [above right=of V] {$0$};
\node[state] (V2) [below right=of V] {$0$};

\path (U) edge (U1);
\path (U) edge (U2);
\path (U1) edge (U11);
\path (U2) edge (U21);
\path (U1) edge (U2);

\path (V) edge (V1);
\path (V) edge (V2);

\end{tikzpicture}
}
\caption{Betweenness Centrality does not satisfy Axiom \ref{ax:struc}. Again the centrality values of the nodes are labeled with the nodes. According to Axiom \ref{ax:struc}, the green node on the left component should have been strictly more central than that of the green node on the right component.}
\label{fig:BC}
\end{figure}

\begin{lemma}
Betweenness Centrality does not satisfy Axiom \ref{ax:struc}.
\end{lemma}
\begin{proof}
Again we construct a counterexample as depicted in Figure \ref{fig:BC}. Betweenness centrality of each of the nodes is marked on the node. If we compare the nodes with green color, they have same number of neighbors and the betweenness centrality of each of the neighbors of the left green colored node is more than the corresponding neighbor of the right green colored node. But the betweenness centrality of the right green colored node is higher than that of the left green colored node. Hence proved.
\end{proof}

\subsection{Weighted Degree Centrality}
Let us define another Centrality measure, namely Weighted Degree Centrality (WDC) as:
\begin{align}
	WDC_u(G) = \sum\limits_{w \neq u; w \in V} \frac{Deg(w)}{dist(u,w)}
	\label{eq:WDC}
\end{align}
Clearly it is a generalization of the degree centrality. WDC of a node depends upon all other nodes in the network. As the distance of a node increases from the node under consideration, its contribution also decreases. Hence, WDC of a node captures the global effect of all other nodes on the node under consideration and hence it is expected to perform better than degree centrality which captures just the degree of the node.

\begin{lemma}
Weighted degree centrality satisfies Axioms \ref{ax:iso}, \ref{ax:loc}, \ref{ax:min} and \ref{ax:mon}.
\end{lemma}
We skip the proof as it is immediate.

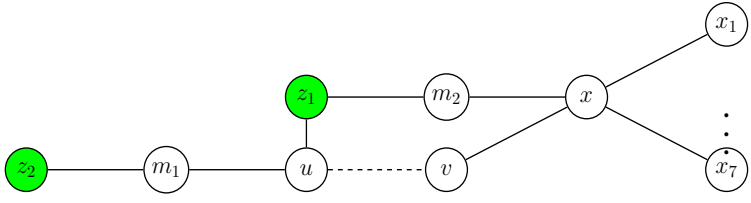
\begin{figure}
\centering
\resizebox{3.9in}{1in}{
\begin {tikzpicture}[auto, node distance =1.5 cm and 3cm, on grid, every loop/.style={},thick, state/.style={circle,draw,minimum size=3em,font=\sffamily\Large\bfseries}]

\node[state] (U) {$u$};
\node[state, fill=green] (Z1) [above=of U] {$z_1$};
\node[state] (V) [right=of U] {$v$};
\node[state] (M1) [left=of U] {$m_1$};
\node[state] (M2) [right=of Z1] {$m_2$};
\node[state, fill=green] (Z2) [left=of M1] {$z_2$};
\node[state] (X) [right=of M2] {$x$};
\node[state] (X1) [above right=of X] {$x_1$};
\node[state] (X7) [below right=of X] {$x_7$};

\path (U) edge (Z1);
\path (Z1) edge (M2);
\path (M2) edge (X);
\path (V) edge (X);
\path (U) edge (M1);
\path (M1) edge (Z2);
\path (X) edge (X1);
\path (X) edge (X7);

\path (U) edge (V) [dashed];

\path (X7) -- (X1) node [black, font=\Huge,  midway, sloped] {$\dots$};

\end{tikzpicture}
}
\caption{WDC does not satisfy Axiom \ref{ax:hop}. Change of centrality for $z_2$ is more that that for $z_1$ because of adding the new edge $\{u,v\}$.}
\label{fig:WDCHop}
\end{figure}

\begin{lemma}
Weighted Degree Centrality does not satisfy axiom \ref{ax:hop}.
\end{lemma}
\begin{proof}
We give a counterexample where WDC fails to satisfy Axiom \ref{ax:hop}. In Figure \ref{fig:WDCHop}, suppose $G$ is the graph without the edge $\{u,v\}$ and $G'$ is the same graph with the added edge $\{u,v\}$.

As $z_1 \in H_G^1(u,v)$ and $z_2 \in H_G^2(u,v)$, according to Axiom \ref{ax:hop} $WDC_{z_1}(G') - WDC_{z_1}(G)$ $>$ $WDC_{z_2}(G') - WDC_{z_2}(G)$. Now, $WDC_{z_1}(G') - WDC_{z_1}(G)$ $=$ $(3/1 + 2/2) - (2/1 - 1/3)$ $=$ $5/3$.

But, $WDC_{z_2}(G') - WDC_{z_2}(G)$ $=$ $(3/2 + 2/3 + (7+2)/4 + 7 \times 1/5) - (2/2 + 1/6 + (7+2)/5 + 7 \times 1/6)$ $=$ $1 + 41/60$ $>$ $WDC_{z_1}(G') - WDC_{z_1}(G)$.
\end{proof}

\begin{figure}
\centering
\resizebox{4in}{1.2in}{
\begin {tikzpicture}[auto, node distance =1 cm and 3cm, on grid, every loop/.style={},thick, state/.style={circle,draw,minimum size=3em,font=\sffamily\Large\bfseries}]

\node[state, fill=green] (U) {$u$};
\node[state] (U1) [above left=of U] {$u_1$};
\node[state] (U2) [below left=of U] {$u_2$};
\node[state] (U3) [below left=of U1] {$u_3$};
\node[state] (U4) [above left=of U3] {$u_4$};
\node[state] (U5) [below left=of U3] {$u_5$};

\node[state, fill=green] (V) [right=of U] {$v$};
\node[state] (V1) [above right=of V] {$v_1$};
\node[state] (V2) [below right=of V] {$v_2$};
\node[state] (V3) [above right=of V1] {$v_3$};
\node[state] (V4) [right=of V1] {$v_4$};
\node[state] (V5) [right=of V2] {$v_5$};
\node[state] (V6) [below right=of V2] {$v_6$};

\path (U) edge (U1);
\path (U) edge (U2);
\path (U1) edge (U3);
\path (U2) edge (U3);
\path (U3) edge (U4);
\path (U3) edge (U5);
\path (V) edge (V1);
\path (V) edge (V2);
\path (V1) edge (V3);
\path (V1) edge (V4);
\path (V2) edge (V5);
\path (V2) edge (V6);

\end{tikzpicture}
}
\caption{WDC does not satisfy Axiom \ref{ax:struc}. Here, $WDC_{u_1} > WDC_{v_1}$ and $WDC_{u_2} > WDC_{v_2}$, but $WDC_{v} > WDC_u$.} 
\label{fig:WDCStruc}
\end{figure}
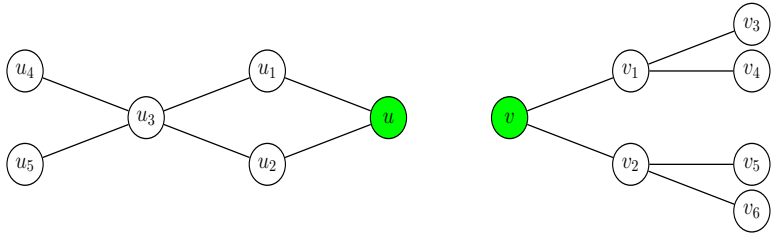

\begin{lemma}
Weighted degree centrality does not satisfy Axiom \ref{ax:struc}.
\end{lemma}
\begin{proof}
Again we construct a counterexample to prove it. Clearly, $WDC_{u_1}$  $= $ $1 \times 2 + 1 \times 4 + 1/2 \times 2 + 1/2 \times 1 + 1/2 \times 1$ $=$ $8$. So from the structural symmetry, $WDC_{u_1} = 8$.

Similarly, one can check, $WDC_{v_1}$ $=$ $WDC_{v_2}$ $=$ $37/6$. So, $WDC_{u_1} > WDC_{v_1}$ and $WDC_{u_2} > WDC_{v_2}$ and hence from Axiom \ref{ax:struc}, $F_{u} > F_{v}$.

Now, $WDC_{u}$ $=$ $2 \time 1/1 \times 2 + 1/2 \times 4 + 2 \times 1/3 \times 1$ $=$ $20/3$, and similarly, $WDC_{v}=8$ $> WDC_u$. This is a contradiction.
\end{proof}

\subsection{Eigenvector Centrality}

Eigenvector Centrality (EC) is defined as:
\begin{align}
	EC_u(G) = x_u
\end{align}
where $x_u$ is the $u$-th component of the eigenvector $x$ corresponding to the maximum eigenvalue of the adjacency matrix of the graph $G$. Throughout the rest of the paper, we assume that the eigenvectors are 2-norm normalized (unit vectors).

\begin{lemma}
Eigenvector centrality satisfies Axiom \ref{ax:iso} and \ref{ax:min}.
\end{lemma}
Again we skip the proof as it is straight forward.

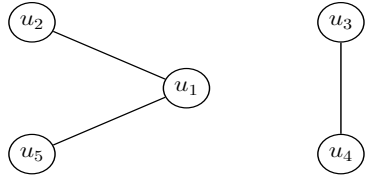
\begin{figure}[t!]
\centering
\resizebox{1.9in}{0.9in}{
\begin {tikzpicture}[auto, node distance =1.5 cm and 3cm, on grid, every loop/.style={},thick, state/.style={circle,draw,minimum size=3em,font=\sffamily\Large\bfseries}]

\node[state] (U1) {$u_1$};
\node[state] (U2) [above left=of U1] {$u_2$};
\node[state] (U3) [below left=of U1] {$u_5$};
\node[state] (U4) [above right=of U1] {$u_3$};
\node[state] (U5) [below right=of U1] {$u_4$};

\path (U1) edge (U2);
\path (U1) edge (U3);
\path (U4) edge (U5);

\end{tikzpicture}
}
\caption{Eigenvector Centrality does not satisfy Axiom \ref{ax:loc}. $EC_{u_3}(G) \neq EC_{u_3}(G[K_{u_3}(G)])$ and $EC_{u_4}(G) \neq EC_{u_4}(G[K_{u_4}(G)])$.}
\label{fig:EVloc}
\end{figure}








\begin{lemma}
Eigenvector Centrality does not satisfy Axiom \ref{ax:loc}.
\end{lemma}
\begin{proof}
Consider the graph in Figure \ref{fig:EVloc}. The    eigenvector centrality of the first graph is [0.7071, 0.5, 0.5, 0, 0], and the maximum eigenvalue is 1.4142. The same centrality measure for the component containing $\{u_1, u_2, u_3\}$ is [0.7071, 0.5, 0.5], but for the component containing $\{u_3, u_4\}$, it is [0.7071, 0.7071]. Hence it violates Axiom \ref{ax:loc}.
\end{proof}

\begin{figure}
\center
\resizebox{1.3in}{1in}{
\begin {tikzpicture}[auto, node distance =1.5 cm and 3cm, on grid, every loop/.style={},thick, state/.style={circle,draw,minimum size=3em,font=\sffamily\Large\bfseries}]

\node[state, fill=green] (U) {$u$};
\node[state] (V) [above left=of U] {$v$};
\node[state, fill=green] (W) [below left=of U] {$w$};

\path (U) edge (V);
\path (U) edge (W)[dashed];

\end{tikzpicture}
}
\caption{Eigenvector Centrality does not satisfy Axiom \ref{ax:mon}. Adding the new edge $\{u,w\}$ does not increase the centrality of the node $u$.}
\label{fig:ECmon}
\end{figure}
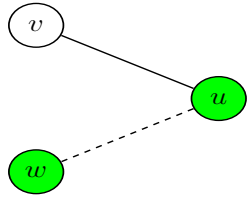

\begin{lemma}
Eigenvector centrality does not satisfy Axiom \ref{ax:mon}.
\end{lemma}
\begin{proof}
We give a counterexample as shown in Figure \ref{fig:ECmon}. Again the graph $G'$ is obtained from the graph $G$ by adding the new edge $\{u,w\}$. But $EC(G)=[0.7071\;, 0.7071\;, 0]$ (in the order of $u$, $v$ and $w$) and $EC(G')=[0.7071\;, 0.5\;, 0.5]$. Hence $EC_u(G) = EC_u(G')$, which violets Axiom \ref{ax:mon}.
\end{proof}

\begin{lemma}
Eigenvector Centrality does not satisfy Axiom \ref{ax:hop}.
\end{lemma}
\begin{proof}
Please see Figure \ref{fig:EVhop}. The (normalized) 
eigenvector centrality of the graph before adding the edge is [0.4647 0.5573 0.2610 0.4647 0.4352]$^T$, whereas the same after adding the edge $\{u_2, u_5\}$ is [0.4119 0.5825 0.2169 0.4119 0.5237]$^T$. Clearly $|EC_{u_2}(G') - EC_{u_2}(G)| = 0.0252$, but $|EC_{u_1}(G') - EC_{u_1}(G)| = 0.0528$ and hence $|EC_{u_1}(G') - EC_{u_1}(G)| > |EC_{u_2}(G') - EC_{u_2}(G)|$. This is a contradiction since $u_2 \in H_G^0(u_2,u_5)$ and $u_1 \in H_G^1(u_2,u_5)$.
\end{proof}

\begin{lemma}
Eigenvector Centrality satisfies Axiom \ref{ax:struc}.
\end{lemma}
\begin{proof}
According to the definition, eigenvector centrality is $x$ where,
\begin{flalign} \label{eq:EVE}
    Ax = \lambda_m x
\end{flalign}
$A \in \mathbb{R}^{n \times n}$ is the adjacency matrix of the given graph $G$, $\lambda_m$ is the maximum eigenvalue of $A$, and $x \in \mathbb{R}$ is a vector whose components are eigenvector centrality of the nodes of $G$. Now to prove that it satisfies Axiom \ref{ax:struc}, suppose there are two nodes $u$ and $v$ such that $|N_G(u)| \geq |N_G(v)|$ and there exists a subset $\bar{N}_G(u) \subseteq N_G(u)$ with $|\bar{N}_G(u)| = |N_G(v)|$, such that there is a bijection $h$ which attaches each vertex $a \in \bar{N}_G(u)$ to a unique vertex $h(a) \in N_G(v)$ so that $EC_a(G) > EC_{h(a)}(G)$. Clearly in that case, $\sum\limits_{w \in N_G(u)} EC_w(G) > \sum\limits_{w \in N_G(v)} EC_w(G)$. 

But according to the Equation \ref{eq:EVE}, $EC_u(G) = \lambda_m \times \sum\limits_{w \in N_G(u)} EC_w(G) $ $>$ $\lambda_m \times \sum\limits_{w \in N_G(v)} EC_w(G)$ $=$ $EC_v(G)$. Hence proved.
\end{proof}

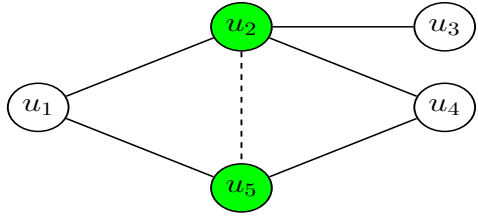
\begin{figure}
\center
\resizebox{2.5in}{1.1in}{
\begin {tikzpicture}[auto, node distance =1.5 cm and 3cm, on grid, every loop/.style={},thick, state/.style={circle,draw,minimum size=3em,font=\sffamily\Large\bfseries}]

\node[state] (U1) {$u_1$};
\node[state, fill=green] (U2) [above right=of U1] {$u_2$};
\node[state, fill=green] (U5) [below right=of U1] {$u_5$};
\node[state] (U3) [right=of U2] {$u_3$};
\node[state] (U4) [below=of U3] {$u_4$};

\path (U1) edge (U2);
\path (U1) edge (U5);
\path (U2) edge (U3);
\path (U2) edge (U4);
\path (U4) edge (U5);

\path (U2) edge (U5) [dashed];

\end{tikzpicture}
}
\caption{Eigenvector Centrality does not satisfy Axiom \ref{ax:hop}. The increase of eigenvector centrality in $u_1$ is more than the increase of that in $u_2$ because of adding the new edge $\{u_2,u_5\}$.}
\label{fig:EVhop}
\end{figure}

\subsection{Decaying Degree Centrality}
Let us define a new centrality measure, Decaying Degree Centrality (DDC) as:
\begin{align}
	DDC_u(G) = \sum\limits_{w \in V} \frac{Deg(w)}{n ^ {2 \times dist(u,w)}}
	\label{eq:DDC}
\end{align}
Note that DDC is also a generalization of degree centrality. But compared to the weighted degree centrality, contribution of a node to the centrality of the node under consideration decreases exponentially with an increase in the distance between them. We characterize DDC with respect to our axiomatic framework in the form of the following lemmas.

\begin{lemma}
DDC satisfies Axioms \ref{ax:iso}, \ref{ax:min} and \ref{ax:mon}.
\end{lemma}
We skip the proof as it is trivial to check.

\begin{lemma}
DDC does not satisfy Axiom \ref{ax:loc}.
\end{lemma}
\begin{proof}
One can check that, as the total number of nodes $n$ in the graph comes in the denominator in the definition of DDC in equation \ref{eq:DDC}, it violates Axiom \ref{ax:loc} when considering the subgraph containing just the respective component.
\end{proof}

\begin{lemma}
DDC satisfies Axiom \ref{ax:hop}.
\end{lemma}
\begin{proof}
As usual, $G=(V,E)$, and $G'=(U,V \cup \{u,v\}$ where $\{u,v\} \notin E$. Suppose $\exists h, \bar{h} \in \mathbb{N}$, s.t., $h < \bar{h}$ and $\exists$ $w_h \in H_G^h(u,v)$ and $w_{\bar{h}} \in H_G^{\bar{h}}(u,v)$. Hence, $|DDC_{w_h}(G') - DDC_{w_h}(G)| \geq \frac{1}{n^{2.h}}$. Now, \\
\begin{flalign*}
&|DDC_{w_{\bar{h}}}(G') - DDC_{w_{\bar{h}}}(G)|&\\
&\leq \frac{x_{\bar{h}}}{n^{2.\bar{h}}} + \frac{x_{\bar{h}+1}}{n^{2.(\bar{h}+1)}} + \cdots &\\
&\text{[$\because$ No change is possible for the nodes below $\bar{h}$-hop neighbors,}&\\
&\text{and $x_{\bar{h}}$ is the sum of the degrees of the nodes which moves to} & \\
&\text{the $\bar{h}$-hop neighbor due to the addition of the new edge $\{u,v\}$]}&\\
&\leq \frac{x_{\bar{h}} + x_{\bar{h}+1} + \cdots}{n^{2.\bar{h}}} &\\
&< \frac{n^2}{n^{2.\bar{h}}} \;\; \leq \frac{1}{n^{2.h}} \;\; \leq |DDC_{w_h}(G') - DDC_{w_h}(G)|&
\end{flalign*}
\end{proof}

\begin{lemma} \label{lemma:DDC}
For any two nodes $u$ and $v$ $\in V$, $DDC_u(G) > DDC_v(G)$ if and only if there is a non-negative integer $h$ such that $\sum\limits_{w \in H_G^h(u)} degree(w) > \sum\limits_{w \in H_G^h(v)} degree(w)$ and $\sum\limits_{w \in H_G^{h'}(u)} degree(w) = \sum\limits_{w \in H_G^{h'}(v)} degree(w)$, $\forall h' =$ \\ $0,1,\cdots,h-1$
\end{lemma}
\begin{proof}
'$\Rightarrow$': 
We use proof by contradiction. Assume there is a non-negative integer $h$ such that $\sum\limits_{w \in H_G^h(u)} degree(w) < \sum\limits_{w \in H_G^h(v)} degree(w)$ and \\
$\sum\limits_{w \in H_G^{h'}(u)} degree(w) = \sum\limits_{w \in H_G^{h'}(v)} degree(w)$, $\forall h' = 0,1,\cdots,h-1$. Now,
\begin{flalign*}
&DDC_u(G) - DDC_v(G)&\\
&= \sum\limits_{h' \in \{0,1,\cdots\} } \sum\limits_{w \in H_G^{h'}(u)} \frac{degree(w)}{n^{2 \times h'}} - \sum\limits_{h' \in \{0,1,\cdots\} } \sum\limits_{w \in H_G^{h'}(v)} \frac{degree(w)}{n^{2 \times h'}}&\\
&= \frac{\sum\limits_{w \in H_G^h(u)} degree(w) - \sum\limits_{w \in H_G^h(v)} degree(w)}{n^{2 \times h}} + \sum\limits_{h' > h} \sum\limits_{w \in H_G^{h'}(u)} \frac{degree(w)}{n^{2 \times h'}}&\\
& \;\;\;\;\;\;\;\;\;\;\;\;\;\;\;\;\;\;\;\;\;\;\;\;\;\;\;\;\;\;\;\;\;\;\;\;\;\;\;\;\;\;\;\;\;\;\;\;\;\;\;\;\;\;\;\;\;\;\;\;\;\;\;\;\;\;\;\;\;\;\;\;\;\; - \sum\limits_{h' > h} \sum\limits_{w \in H_G^{h'}(v)} \frac{degree(w)}{n^{2 \times h'}}&\\
&< \frac{\sum\limits_{w \in H_G^h(u)} degree(w) - \sum\limits_{w \in H_G^h(v)} degree(w)}{n^{2 \times h}} - \sum\limits_{h' > h} \sum\limits_{w \in H_G^{h'}(v)} \frac{degree(w)}{n^{2 \times h'}}&\\
&\leq \frac{-1}{n^{2 \times h}} - \sum\limits_{h' > h} \sum\limits_{w \in H_G^{h'}(v)} \frac{degree(w)}{n^{2 \times h'}}&\\
&<\frac{-1}{n^{2 \times h}} - \frac{n^2}{n^{2 \times (h+1)}} \;\; = \;\; 0&
\end{flalign*}
which is a contradiction.

'$\Leftarrow$': The converse can be proved in the same way.
\end{proof}

\begin{figure}
\center
\resizebox{4in}{1.6in}{
\begin {tikzpicture}[auto, node distance =1.5 cm and 3cm, on grid, every loop/.style={},thick, state/.style={circle,draw,minimum size=3em,font=\sffamily\Large\bfseries}]

\node[state, fill=green] (U) {$u$};
\node[state] (U1) [above left=of U] {$u_1$};
\node[state] (U2) [below left=of U] {$u_2$};
\node[state] (U11) [left=of U1] {$u_{11}$};
\node[state] (U21) [left=of U2] {$u_{22}$};

\node[state, fill=green] (V) [right=of U] {$v$};
\node[state] (V1) [above right=of V] {$v_1$};
\node[state] (V2) [below right=of V] {$v_2$};
\node[state] (V11) [above right=of V1] {$v_{11}$};
\node[state] (V12) [right=of V1] {$v_{12}$};
\node[state] (V21) [right=of V2] {$v_{21}$};
\node[state] (V22) [below right=of V2] {$v_{22}$};

\path (U) edge (U1);
\path (U) edge (U2);
\path (U1) edge (U11);
\path (U2) edge (U21);
\path (U1) edge (U2);
\path (U11) edge (U21);

\path (V) edge (V1);
\path (V) edge (V2);
\path (V1) edge (V11);
\path (V1) edge (V12);
\path (V2) edge (V21);
\path (V2) edge (V22);

\end{tikzpicture}
}
\caption{DDC does not satisfy Axiom \ref{ax:struc}. Here, $DDC(u_1)>DDC(v_1)$ and $DDC(u_2)>DDC(v_2)$, but $DDC_u = DDC_v$.}
\label{fig:DDC}
\end{figure}
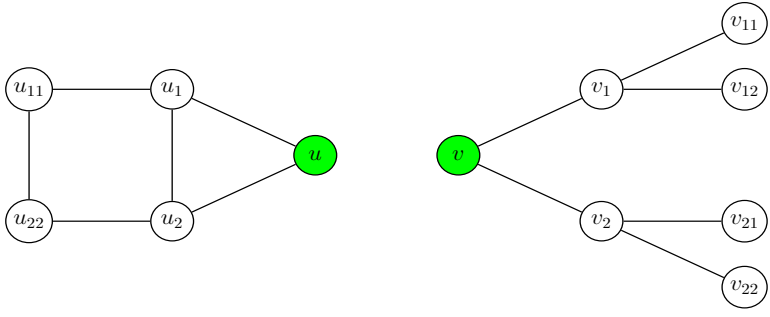

\begin{lemma}
DDC does not satisfy Axiom \ref{ax:struc}.
\end{lemma}
\begin{proof}
Consider the network in Figure \ref{fig:DDC}. Clearly $degree(u_1) = degree(u_2) = degree(v_1) = degree(v_2) = 3$. But the sum of the degrees of nodes in the 1-hop neighborhood of $u_1$ and $u_2$ are 7 each, but the same for $v_1$ and $v_2$ are 4 each. Hence, $DDC(u_1)>DDC(v_1)$ and $DDC(u_2)>DDC(v_2)$, from Lemma \ref{lemma:DDC}. But sum of the degrees of the nodes in $h$-hop neighborhood of $u$ is same to that in the $h$-hop neighborhood of $v$, $\forall h=0,1,2$ and hence $DDC_u = DDC_v$. Hence DDC does not satisfy Axiom \ref{ax:struc}.
\end{proof}

\begin{table*}
    \caption{Axioms Satisfied by Different Centrality Measures}
	\centering
	\begin{adjustbox}{width=\textwidth}
	\begin{tabular}{*8c}
	\toprule
	Centrality Measures & Axiom 1 & Axiom 2 & Axiom 3 & Axiom 4 & Axiom 5 & Axiom 6 \\ \hline
	\midrule
	Uniform Centrality & \cmark & \cmark & \xmark & \xmark & \xmark & \xmark \\
	Degree Centrality & \cmark & \cmark & \cmark & \cmark & \xmark & \xmark \\
	Closeness Centrality & \cmark & \cmark & \cmark & \cmark & \xmark & \xmark \\
	Betweenness Centrality & \cmark & \cmark & \cmark & \xmark & \xmark & \xmark \\
	Weighted Degree Centrality & \cmark & \cmark & \cmark & \cmark & \xmark & \xmark \\
	Eigen Vector Centrality & \cmark & \xmark & \cmark & \xmark & \xmark & \cmark \\
	Decaying Degree Centrality & \cmark & \xmark & \cmark & \cmark & \cmark & \xmark \\
	\bottomrule
	\end{tabular}
	\end{adjustbox}
	\label{tab:sat}

	\end{table*} 
	
\section{Discussion and Future Work}
In this paper we proposed an axiomatic framework for the centrality measures for networks and analyzed some fundamental measures of centrality with respect to the proposed framework. Satisfiability or otherwise of  different axioms by different centrality measures are summarized in Table \ref{tab:sat}. Following are the key observations made from the last two sections.
\begin{itemize}
    \item We have proposed six axioms in total to capture different intrinsic properties of a centrality measure. They can be used to develop new centrality measures and for partial ranking of the existing measures. Though each axiom captures only a basic property of a centrality measure, surprisingly many well-known existing centrality measures could not satisfy many of them. We also proposed some new centrality measures in the process. But they also fail to satisfy all the axioms.
    \item One major contribution of the paper is the analysis of the centrality measures as shown in Table \ref{tab:sat}. But as one can see, degree centrality, closeness centrality and weighted degree centrality, though being significantly different in their definitions, satisfy the same set of axioms. Further investigation is required to understand their role in satisfying the fundamental properties of centrality.
    \item From Table \ref{tab:sat}, we still do not know whether there exist a centrality measure which can satisfy all the axioms of our framework. This would open up the scope of further research in getting some possibility or impossibility results in this direction.
\end{itemize}


\bibliographystyle{spbasic}      
\bibliography{centrality}   

\end{document}